\documentclass[12pt]{article}
\usepackage[breaklinks]{hyperref} 
\usepackage{graphicx}
\usepackage{fullpage}
\usepackage{xspace}
\usepackage{xcolor}
\usepackage{amsmath}
\usepackage{paralist}
\usepackage{amssymb}
\usepackage{amsthm}

\newtheorem{theorem}{Theorem}[section]

\newtheorem{lemma}[theorem]{Lemma}
\newtheorem{claim}[theorem]{Claim}

\renewcommand{\th}{th\xspace}
\newcommand{\atgen}{\symbol{'100}}
\newcommand{\SarielThanks}[1]{\thanks{Department of Computer Science;
      University of Illinois; 201 N. Goodwin Avenue; Urbana, IL,
      61801, USA; {\tt sariel\atgen{}uiuc.edu}; {\tt
         \url{http://www.uiuc.edu/\string~sariel/}.} #1}}
\newcommand{\BernardThanks}{\thanks{Department of Mathematical
      Sciences; University of Illinois; 1409 W. Green St.; Urbana, IL,
      61801, USA; {\tt lidicky\atgen{}illinois.edu}.} }

 \newcommand{\brc}[1]{\left\{ {#1} \right\}}

\newcommand{\pth}[2][\!]{#1\left({#2}\right)}
\newcommand{\cardin}[1]{\left| {#1} \right|}
\renewcommand{\th}{th\xspace}
\newcommand{\CHX}[1]{{\mathcal{CH}}\pth{#1}}

\definecolor{blue25}{rgb}{0.0,0,0.85} \newcommand{\emphic}[2]{
   \textcolor{blue25}{
      \textbf{\emph{#1}}}
   \index{#2}} \newcommand{\emphi}[1]{\emphic{#1}{#1}}

\newcommand{\Grid}[1]{G_{#1}}

\newcommand{\Line}{\ell}

\newcommand{\ceil}[1]{\left\lceil {#1} \right\rceil}
\newcommand{\floor}[1]{\left\lfloor {#1} \right\rfloor}
\newcommand{\Lines}[1]{L_{#1}}

\newcommand{\clmlab}[1]{\label{claim:#1}}
\newcommand{\clmref}[1]{Claim~\ref{claim:#1}}

\newcommand{\figlab}[1]{\label{fig:#1}}
\newcommand{\figref}[1]{Figure~\ref{fig:#1}}

\newcommand{\lemlab}[1]{\label{lemma:#1}}
\newcommand{\lemref}[1]{Lemma~\ref{lemma:#1}}

\newcommand{\PntSet}{P}
\newcommand{\VecSet}[1]{\mathcal{V}_{#1}}

\newcommand{\Barany}{B{\'a}r{\'a}ny\xspace}

\newcommand{\pnt}{\mathsf{p}}
\newcommand{\depthP}[2]{\mathsf{d}_{#1}\pth{#2}}
\newcommand{\Vertices}[1]{V \pth{#1}}

\begin{document}

\title{Peeling the Grid}

\author{
   Sariel Har-Peled\SarielThanks{Work on this paper was partially
      supported by a NSF AF award CCF-0915984.}
   \and
   Bernard Lidick\'{y}\BernardThanks{}
}

\date{\today}

\maketitle

\begin{abstract}
    Consider the set of points formed by the integer $n \times n$
    grid, and the process that in each iteration removes from the
    point set the vertices of its convex-hull. Here, we prove that the
    number of iterations of this process is $O\pth{n^{4/3}}$; that is,
    the number of convex layers of the $n\times n$ grid is
    $\Theta\pth{n^{4/3}}$.
\end{abstract}

\newpage

\section{Introduction}

For many algorithms, the worst case behavior is rarely encountered in
practice. This is because the worst case behavior might require a
degenerate and convoluted input. To address this gap between the worst
case analysis and a real world behavior, a considerable amount of
research was spent on analyzing algorithms and discrete geometric
structures under certain assumptions on the input, including
\begin{inparaenum}[(i)]
    \item realistic input models \cite{bksv-rimga-02},
    \item fatness \cite{abes-ibufop-11},
    \item randomness, etc.
\end{inparaenum}

\begin{figure}
    \centerline{
       \begin{tabular}{ccc}
           \qquad
           {\includegraphics[page=1]{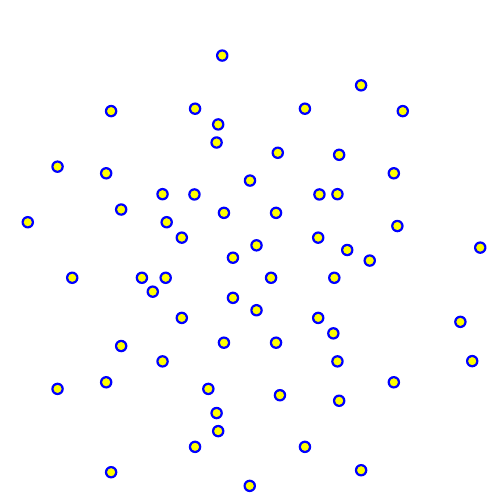}}
           \qquad
           &
           \qquad
           \qquad
           &
           \includegraphics[page=2]{layers}
       \end{tabular}}
    \caption{A point set and its decomposition into convex layers.}
    \figlab{layers}
\end{figure}

\paragraph{Random points.}

There is a significant amount of work on the geometric behavior of
random point sets \cite{rs-udkhv-63, r-slcdn-70, ww-sg-93, b-rplpc-08,
   obsc-stcav-00, jn-osdpv-04}. The question of how the Voronoi
diagram or the convex-hull of a point set randomly generated inside a
convex domain behaves had received considerable attention. In
particular, it is known that for a set of $n$ points chosen uniformly
in the unit square, the expected complexity of the convex-hull is $O(
\log n)$, and $O(n^{1/3})$ if the domain is a disk (this bound holds
for any convex shape).

\paragraph{Grid points.}
Surprisingly, the known results on uniformly sampled points match the
results known for the grid point set. For example, the number of
vertices of the convex hull of any subset of the $\sqrt{n} \times
\sqrt{n}$ grid is $O\pth{n^{1/3}}$, which matches the bound for the
random points. This phenomena holds for many similar scenarios, see
the survey by \Barany \cite{b-rplpc-08}.

\paragraph{Convex layers.}
The decomposition of a point set into convex layers is one possible
way to measure the depth of a point inside the point set. Formally,
the \emphi{convex depth} of a point $\pnt$ in a point set $\PntSet$ is
$\depthP{\pnt}{\PntSet} = 1$ if $\pnt$ is a vertex of the convex-hull
of $\PntSet$, and it is $\depthP{\pnt}{\PntSet} = 1+
\depthP{\pnt}{\PntSet \setminus \Vertices{\CHX{\PntSet}}}$ otherwise,
where $\CHX{\PntSet}$ denotes the convex-hull of $\PntSet$ and
$\Vertices{\CHX{\PntSet}}$ denotes the set of its vertices
\footnote{A point of $\PntSet$ is a vertex of the convex-hull only if
   it is a corner of the convex-hull. Formally, $\pnt$ is a vertex of
   the convex-hull of $\PntSet$ is $\CHX{\PntSet} \neq \CHX{\PntSet
      \setminus \brc{\pnt}}$.}. 
This partitions the point set into convex-layers, as depicted in
\figref{layers}. In particular, if the points rise out of physical
measurements (that might contain noise), a point with large convex
depth is unlikely to be an outlier. This is one possible definition of
robust statistics for points, although this definition has its
limitations, see \cite{rs-crsdm-04} for details.  In particular,
Chazelle \cite{c-clps-85} provided an $O\pth{n \log n}$ time algorithm
for computing all the convex layers for a set of points in the plane.

For a set of $n$ points picked uniformly inside a bounded convex
domain in $\mathbb{R}^d$, it is known that the expected number of convex layers is
$\Theta\pth{n^{2/(d+1)}}$ \cite{d-co-04}.

\paragraph{Our results.}
In this paper, we are prove that the number of convex layers of the $n
\times n$ grid is $\Theta\pth{n^{4/3}}$. This bound is quite
surprising -- indeed, as demonstrated by~\figref{c:h}, the peeling
process starts out quite slowly, the first three layers having $4, 8,
8$ vertices (independent of the value of $n$), respectively. A priori,
it is not clear why this process accelerates and contains more
vertices. Furthermore, the maximum number of vertices in convex
position in an $n\times n$ grid is $O\pth{n^{2/3}}$ (this is well
known, see \lemref{lower:bound}). Namely, somewhat surprisingly, a
constant fraction of the layers have asymptotically maximum size. Our
result matches the known result for random points. Note, that although
the bounds are similar, the proof for the random point set does not
carry over to the grid case.

We also observe that the number of convex layers is $\Omega(n^2)$ if
the grid of $n \times n$ points is allowed to be non-uniform (instead
of the integer grid used above).  Naturally, in this construction,
where every point is on two lines where each has $n$ points.

\section{Peeling the grid}

Let $\PntSet_0 = \Grid{n} = \brc{1,\ldots,n}^2$, be the $ n\times n$
integer grid. In the $i$\th iteration, consider the convex-hull $C_i =
\CHX{\PntSet_{i-1}}$, for $i=1, \ldots$. Let $V_i$ be the set of
vertices of $C_i$. Naturally, we consider a grid point to be a
\emphi{vertex} only if it is a corner of the convex-hull, and as such
grid points falling in the middle of edges of $C_i$, are not in
$V_i$. Now, let $\PntSet_{i} = \PntSet_{i-1} \setminus V_i$.  In
words, we start with the $n \times n$ grid, and peel away the vertices
of the convex-hull, and we repeat this process till all the grid
points of $\Grid{n}$ are removed.  Let $\tau(n)$ be the number of
iterations, till $\PntSet_i$ is an empty set. Here we are interested
in the behavior of $\tau(n)$. See \figref{c:h} for an example of how
the generated polygons look like.

\begin{figure}[tb]
    \centerline{\includegraphics{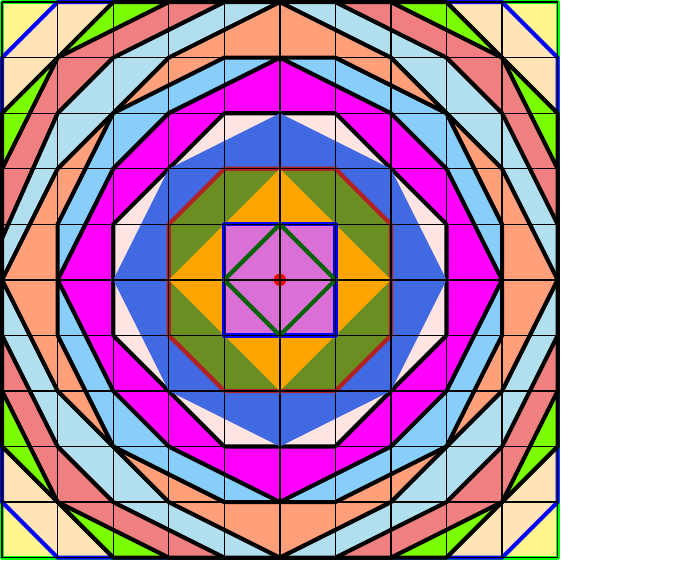}}
    \caption{The polygons generated while peeling the $11\times 11$
       integer grid.}
    \figlab{c:h}
\end{figure}

\subsection{A lower bound on $\tau(n)$}

The following is well known, and we include a proof for the sake of
completeness.

\begin{lemma}
    Given any convex set $C$ in the plane, it can have at most
    $O\pth{n^{2/3}}$ vertices of $\Grid{n}$.
    
    \lemlab{lower:bound}
\end{lemma}
\begin{proof}
    Consider a convex set $C$ such that all its vertices are points of
    $\Grid{n}$. The perimeter of $C$ is at most $4n$. The number of
    edges of the convex hull of $C$ of length at least (or equal to)
    $\mu$ is at most $4n/\mu$. The number of edges having length
    smaller than $\mu$ is bounded by the number of integer points of
    distance at most $\mu$ from the origin, and this number is bounded
    by $(2\mu+1)^2 = O(\mu^2)$. As such, the number of vertices of $C$
    is at most $O( n/\mu + \mu^2)$. Setting $\mu = \floor{ n^{1/3} }$
    then implies the claim.
\end{proof}

As such, $\cardin{V_i} = O\pth{n^{2/3}}$, which implies immediately
that $\tau(n) \geq n^2/ \max_i \cardin{V_i} = \Omega\pth{n^{4/3}}$.

\subsection{An upper bound on $\tau(n)$}

An integer vector $(x,y)$ is \emphi{primitive} if $\gcd(x,y) = 1$.
For an integer $\mu$, let $\VecSet{\mu}$ be the set of all primitive
non-zero

integer vectors $(x,y)$, where $0 \leq y< x \leq \mu$.  The following
is well known, and we sketch a proof for the sake of completeness.

\begin{lemma}
    \lemlab{primitive}
    
    We have $\cardin{\VecSet{\mu}} \geq c \mu^2$, for some constant $c
    > 0$.
\end{lemma}

\begin{proof}
    For a fixed $x$, consider the vectors $(x,y)$ in $\VecSet{\mu}$,
    such that $y < x$, and $\gcd(x,y) = 1$. The number of such vectors
    is the number of integer values of $y$ that are relative prime to
    $x$, and this number is the Euler's totient function $\phi(x)$.
    As such, $\cardin{\VecSet{\mu}} \geq \sum_{i=1}^\mu \phi(i) \geq c
    \mu^2$.  the last step follows from known bounds, see
    \cite{hw-tn-65}.
\end{proof}

In the following, we pick $\mu$ to be smaller than $n/4$, and $n$ is
sufficiently large.

For every vector $v \in \VecSet{\mu}$, consider the set $\Lines{v}$ of
all lines having direction $v$ that intersect the grid points
$\Grid{n}$. Every line in $\Lines{v}$ contains at most
$1+\floor{(n-1)/v_x}$ points of the grid (and most lines in this
family contain at least $\floor{(n-1)/v_x}$ points of the grid (the
only problematic lines are the ones that have short intersection with
the square $[1,n]^2$ because of the corners)).

\begin{figure}[t]
    \begin{tabular}{ccc}
        \includegraphics[page=1,width=0.35\linewidth]{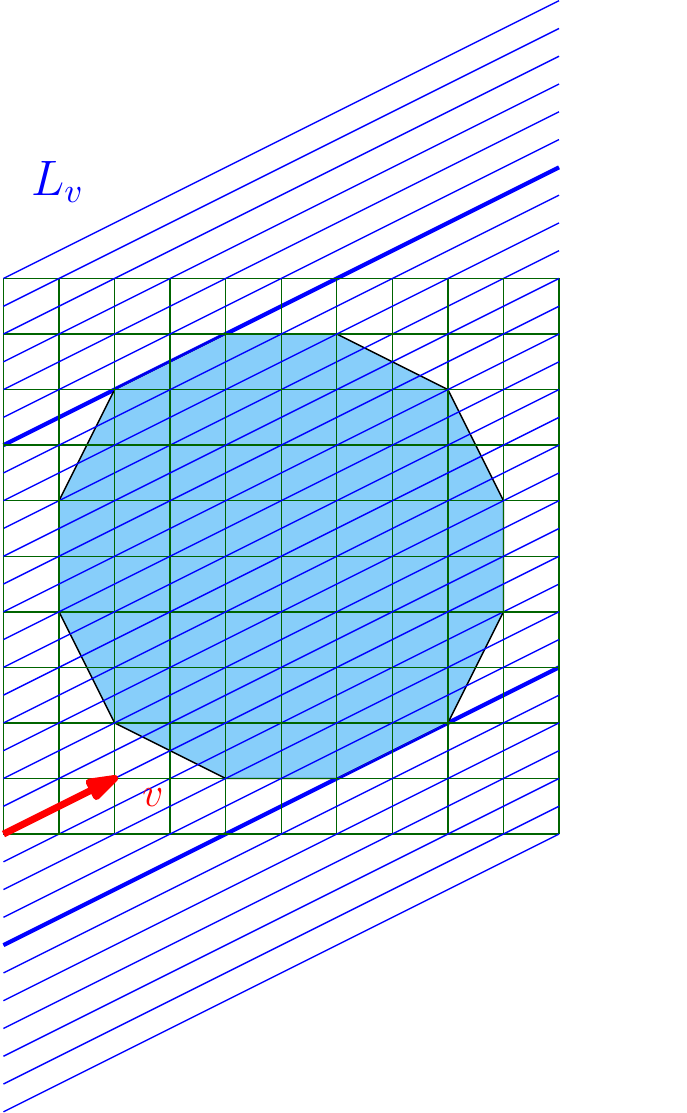}
        &
        \hspace{-1cm}
        \includegraphics[page=2,width=0.35\linewidth]{active}
        &
        \hspace{-1cm}
        \includegraphics[page=3,width=0.35\linewidth]{active}\\
        (A) & 
        \hspace{-1cm} (B) &
        \hspace{-1cm} (C)
    \end{tabular}
    \caption{(A) An active direction $v$, and the set of lines $L_v$.
       (B) An inactive iteration for $v$. (C) The next iteration --
       the two ``old'' tangent lines no longer intersect the current
       convex layer.}
    \figlab{active:inactive}
\end{figure}

\begin{claim}
    \lemlab{active}
    
    For $n>10$, $\mu < n/4$ and $v \in \VecSet{\mu}$, we have that
    $\cardin{\Lines{v}} \leq 4n \mu$.
\end{claim}
\begin{proof}
    A line $\Line \in \Lines{v}$ that intersects $\Grid{n}$ has an
    intersection of length at least $n$ with the enlarged square
    $[1,2n]^2$. Specifically, the projection of the intersection on
    the $x$ axis has length at least $n$. Since $\ell$ has direction
    $v$ and it contains a grid point, it follows that it has grid
    points on it, that are of distance $||v||$ from each other. On the
    projection, the distance between these points is $v_x$. As such,
    this intersection contains at least $1+ \floor{n /v_x} \geq n/\mu$
    points of the grid $\Grid{2n}$ on it. In particular, the number of
    such lines can be at most $4n^2/(n /\mu) = 4n \mu$.
\end{proof}

Since the lines of $\Lines{v}$ cover all the grid points of
$\Grid{n}$, and the vertices of $C_i$ are grid points, it follows that
$\Lines{v}$ always contains two lines that are tangent to $C_i$. If
these two tangent lines intersect $\partial C_i$ along an non-empty
edge, then $v$ is \emphi{active} at iteration $i$ (i.e., $v$ is not
active if the two tangents touch $C_i$ at a vertex).

In the following, we slightly abuse notations and use $\Lines{v} \cap
C_i$ to denote the set of all lines of $\Lines{v}$ that have non-empty
intersection with $C_i$. 

\begin{claim}
    \clmlab{obvious}
    
    If $v$ is not active at iteration $i$, then $\cardin{\Lines{v}
       \cap C_{i+1}} \leq \cardin{\Lines{v} \cap C_i} - 2$.
\end{claim}

\begin{proof}
    If $v$ is not active at iteration $i$ then a tangent $\Line$ to
    $C_i$ from $\Lines{v}$ intersects $C_i$ only at a vertex. But this
    vertex is being removed from the point set when computing
    $\PntSet_{i+1}$. In particular, the line $\Line$ no longer
    intersects $C_{i+1}$. The same argument also applies to the other
    tangent. This is demonstrated in \figref{active:inactive}.
\end{proof}

\begin{claim}
    Throughout the process, for a vector $v \in \VecSet{\mu}$, it can
    be inactive in at most $2n \mu$ iterations.
\end{claim}
\begin{proof}
    Every time $v$ is not active, the number of lines of $\Lines{v}$
    that intersect the active convex-hull decreases by two, by
    \clmref{obvious}. By \lemref{active} there are at most $4n \mu$
    lines in the set $\Lines{v}$, and as such this can happen at most
    $4n \mu / 2$ times.
\end{proof}

If the process continues more than $M = 4n \mu$ iterations then every
vector in $\VecSet{\mu}$ is active in at least half of the
iterations. In particular, if $n_i$ is the number of active directions
at iteration $i$, then we have that
\begin{align*}
    \alpha = \sum_{i=1}^M n_i \geq
    2n \mu \cardin{\VecSet{\mu}}
    \geq
    2c n \mu^3,
\end{align*}
by \lemref{primitive}.

Observe, that if $n_i$ vectors are active at the $i$\th iteration,
then the convex hull of $C_i$ has at least $2n_i$ edges (and thus
vertices) at iteration $i$. As such, if we set $\mu =
\ceil{n^{1/3}/c^{1/3}} = \Theta\pth{n^{1/3}}$, we have that the total
number of vertices of the convex hulls in the first $M$ iterations is
at least
\begin{align*}
    2\alpha \geq 4c n \mu^3 \geq 4n^2,
\end{align*}
which is a contradiction, as the initial grid set has at most $n^2$
points. We conclude that the algorithm must terminate after $M = 4n
\mu = O \pth{n^{4/3}}$ iterations. We thus proved the following.

\begin{theorem}
    Starting with the grid $\Grid{n}$, consider the process that
    repeatedly removes the convex-hull vertices of the current set of
    vertices. This process takes $\Theta\pth{n^{4/3}}$ steps.
\end{theorem}

\section{Lower bound of $\Omega(n^2)$ for a 
   non-uniform grid}

This section is devoted to describing a set $M$ of $n^2$ points in the
plane where the peeling process takes $\Omega\pth{n^2}$ steps.  For
simplicity assume that $n=2k$ for some integer $k$.

Take a collection of $k$ squares $S_1,\ldots,S_k$ where $S_i$ has
length of its side $3^i$ and the squares are positioned such that
their centers coincide with the origin. Let $L$ be the set of $4k$
lines that are obtained by extending the segments of the squares into
lines.  Finally, let $M$ be the set of all intersections of lines in
$L$.  Notice that each line contains $2k$ points and that $|L| = 4k^2
= n^2$.  See \figref{grid:nsquare}.

\begin{figure}
    \centerline{
       \includegraphics[page=1]{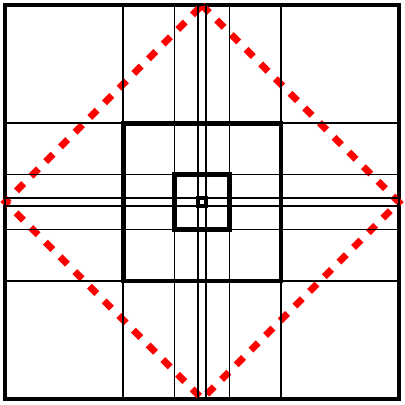}
       \hskip 5em
       \includegraphics[page=1]{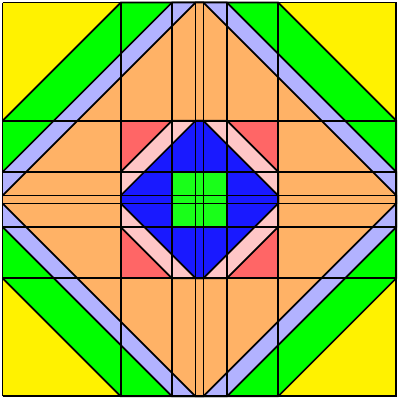}
    }
    \caption{A point set where the peeling process requires
       $\Omega\pth{n^2}$ steps.}
    \figlab{grid:nsquare}
\end{figure}

Let the peeling process partition $M$ into convex sets
$C_1,C_2,\ldots$.
\begin{claim}
    For every $C_i$ exists $S_j$ such that $C_i \subseteq S_j$.
\end{claim}
\begin{proof}
    Let $j$ be the largest index such that $C_i \cap S_j \neq
    \emptyset$.  Notice that $C_i$ is centrally symmetric as $M$ is
    centrally symmetric and this property is preserved by the peeling
    process.  If $|C_i \cap S_j| = 4$ then $C_i \cap S_j$ are the four
    corners of $S_j$ and thus $|C_i| = 4$ as $C_i$ is strictly convex.
    Hence $|C_i \cap S_j| = 8$ and $C_i$ contains points on both
    vertical and horizontal lines of $S_j$ in every quadrant.  Let $D$
    be the square with corners being intersections the axis and
    $\CHX{S_j}$.  See \figref{grid:nsquare} on the left.  Notice that
    $S_l \subset D \subset \CHX{C_i \cap S_j}$ for every $l < j$.
    Therefore $C_i = C_i \cap S_j \subseteq S_j$.
\end{proof}

The previous claim implies that $|C_i| \leq 8$ for every $i$. Hence
the peeling process needs at least $n^2/8 = \Omega\pth{n^2}$ steps.

\section{Conclusions}

The most natural question left by our work, is the prove similar
bounds in higher dimensions. This seems quite challenging, and we
leave it as an open problem for further research.

Let us also note for an interested reader that, according to experiments,
the layers in the peeling process are getting close to circles as the
process is advancing.

\section*{Acknowledgments}

We thank Robert Jamison for many useful discussions on this and
related problems. The authors also would like to thank Imre \Barany
for providing relevant information.

\end{document}